\documentclass[reqno]{amsart}
\usepackage{amssymb, amsmath}
\usepackage{natbib}
\usepackage{lscape}
\usepackage{dsfont}
\usepackage{mathrsfs}
\usepackage{bbm}
\usepackage{graphicx}
\usepackage[pdftex,plainpages=false,colorlinks,hyperindex,bookmarksopen,linkcolor=red,citecolor=blue,urlcolor=blue]{hyperref}

\DeclareMathAlphabet{\mathpzc}{OT1}{pzc}{m}{it}

\bibpunct{[}{]}{;}{n}{,}{,}

\theoremstyle{definition}

\theoremstyle{plain}
\newtheorem{teorema}{Theorem}

\numberwithin{equation}{section}

\allowdisplaybreaks

\def \l { \left( }
\def \r {\right) }
\def \ll { \left\lbrace }
\def \rr { \right\rbrace }

\title[] {MOTION AMONG RANDOM OBSTACLES \\
ON  A HYPERBOLIC SPACE }

\begin{document}

	\author{Enzo Orsingher}
	\author{Costantino Ricciuti} 
    \author{Francesco Sisti}

\address{Dip. di Scienze Statistiche, Sapienza - Universit\`a di Roma , Piazzale Aldo Moro 5, 00185 Roma, Italy.}
	\email{enzo.orsingher@uniroma1.it}

\address{Dip. di Scienze Statistiche, Sapienza - Universit\`a di Roma , Piazzale Aldo Moro 5, 00185 Roma, Italy.}
	\email{costantino.ricciuti@uniroma1.it}

\address{Dip. di Scienze di Base e Applicate per l'Ingegneria, Sapienza - Universit\`a di Roma , Via A. Scarpa 16, 00161 Roma, Italy}
	\email{francesco.sisti@sbai.uniroma1.it}

	\date{\today}
	\subjclass[2000]{TO BE DEFINED}

\begin{abstract}
	We consider the motion of a particle  along the geodesic lines of the Poincar\'e half-plane. The particle is specularly reflected when it hits randomly-distributed obstacles that are assumed to be motionless. This is the hyperbolic version of the well-known Lorentz Process studied  by Gallavotti in the Euclidean context. We analyse the limit in which the density of the obstacles increases to infinity and the size of each obstacle vanishes: under a suitable scaling, we prove that our process converges  to a Markovian process, namely a random flight on the hyperbolic manifold.
	\end{abstract}

	%\address{Department of Statistical Sciences, Sapienza University of Rome}
	%\email {costantino.ricciuti@uniroma1.it}
	\date{\today}
	\subjclass[2000]{TO BE DEFINED}	
	
	\maketitle
	
\tableofcontents

%%%%%%%%%%%%%%%%%%%%%%%%%%%%%%%%%%%%%%%%%%%%%%%%%%%%%%%%%%%%%%%%%%%%%%%%%%%%%%%%%%%%%%%
%%%%%%%%%%%%%%%%%%%%%%%%%%%%%%%%%%%%%%%%%%%%%%%%%%%%%%%%%%%%%%%%%%%%%%%%%%%%%%%%%%%%%%
\section{Introduction}

Since modern physics showed how the Euclidean space is merely a local approximation of the geometry underlying the universe, non Euclidean manifolds has been a topic of growing interest among both mathematicians 
%	\author{
% \Large
%Enzo Orsingher\footnote{Dip. di Scienze Statistiche, ''Sapienza'' Universit\'a di Roma , Piazzale Aldo Moro 5, 00185 Roma, Italy. Email: enzo.orsingher@uniroma1.it},  \, Costantino Ricciuti\footnote{Dip. di Scienze Statistiche, ''Sapienza'' Universit\'a di Roma , Piazzale Aldo Moro 5, 00185 Roma, Italy. Email: costantino.ricciuti@uniroma1.it} \,Francesco Sisti\footnote{ Dip. di Scienze di Base e Applicate per l'Ingegneria, ''Sapienza'' Universit\'a di Roma , Via A. Scarpa 16, 00161 Roma, Italy. Email: francesco.sisti@sbai.uniroma1.it}     
%    
and  physicists.

% Indeed every dynamic taking place on a non eucliedean space will be invariably affected by the metric of the space itself.
 %**** Random motions in non-Euclidean spaces have attracted the interest of many mathematicians in the last decades. The main reason of this interest  lies in the fact that  modern physics states that the universe is governed by non euclidean geometries. In particular, random models in hyperbolic spaces have a prominent role in the literature. ****
 
Over the last few decades much attention has been paid to the random motions that occur in non-Euclidean spaces.  In particular, random motions on hyperbolic spaces are often reported in the literature.
Most of the papers focus on hyperbolic Brownian motion, but   random motions at finite velocity have also been  investigated recently (see \cite{cammarota1, cammarota2, cammarota3,degregorio}).

We believe that the Lorentz process, studied in detail by Gallavotti \cite{Gallavotti}, is a paradigmatic model for finite velocity random motions and therefore we have focused our attention on this topic.

The Lorentz process studied by Gallavotti is the motion performed by a particle  in the Euclidean plane $\mathbb{R}^2$, where static spherical obstacles are distributed according to a Poisson probability measure with intensity $\lambda$. The particle is assumed to move along straight lines and to be specularly reflected by the obstacles. The process is clearly non-Markovian, because the trajectories recall the effect of previous collisions. 
The main result of Gallavotti's study is the proof of  consistency of the so-called Boltzmann-Grad limit, in which the radius $r$ of each obstacle decreases to zero and  the density $\lambda $ increases to infinity, so that the mean free path $(2 \lambda r)^{-1}$ remains constant (few collisions regime).   
 Under the Boltzmann-Grad asymptotics, the single-time probability density of the Lorentz process converges to that of a Markovian process, solving a linear Boltzmann equation.

Gallavotti's model  was firstly improved by Spohn \cite{Spohn} and Boldrighini et al. \cite{boldrighini}, and was further developed in subsequent papers. In  \cite{Desvillettes},  Desvillettes and Ricci  studied a variant of the model in which the obstacles are totally absorbing and an external force field is present; in this case the Boltzmann-Grad limit does not lead to a Markovian process, unless a random motion of the obstacles is assumed (with Gaussian distribution of velocities). 
Gallavotti's work has also inspired the approach of  Basile et al. \cite{Basile},
for which there is a  slightly different context, especially because the obstacles are  circular potential barriers instead of being hard spheres.  However, the authors followed in Gallavotti's footsteps and, by suitably scaling the Poissonian density and  the potential intensity, they obtained a Markovian approximation that is governed by a linear Landau equation; the limiting process is halfway between the Boltzmannn-Grad regime (few collisions) and the weak-coupling regime (many collisions).

In this paper we construct the Lorentz process in the Poincar\'e half-plane, which is one of the most popular Euclidean models of hyperbolic spaces.  
For the reader's convenience, we will report some of the basic facts on the Poincar\'e half-plane. It can be defined as the region $\mathbb{H}_2=\{(x,y) \in \mathbb{R}^2 \, :\, y >0\}$ endowed with the metric
\begin{align} \label{metrica iperbolica}
ds ^2= \frac{dx^2+dy^2}{y^2}.
\end{align}

From (\ref{metrica iperbolica}) it follows that  geodesic curves   in $\mathbb{H}_2$ are either Euclidean half-circumferences with their centers on the $x$-axis or  Euclidean lines parallel to the y-axis. Moreover, the hyperbolic distance  between two points  $P_1=(x_1,y_1)$ and $ P_2=(x_2,y_2) $  in $\mathbb{H}_2$, denoted as $d_h(P_1,P_2)$, is given by
\begin{align} \label{distanza tra due punti}
\cosh \bigl ( d_h(P_1,P_2) \bigl ) = \frac{(x_1-x_2)^2 +y_1^2+y_2^2}{2y_1y_2}.
\end{align} 
An  hyperbolic circle with center $C=(x_c, y_c)$ and hyperbolic radius $\eta$ is defined as $ B_{\eta}(C)= \ll  w\in \mathbb{H}_2 : \, d_h(w,C)\leq \eta \rr $ and its border is denoted as $\partial B_{\eta}(C)$; from (\ref{distanza tra due punti}) we obtain that:
\begin{align} \label{hyperbolic circle}
\partial B_{\eta}(C)= \{(x,y) \in \mathbb{H}_2 :\, (x-x_c)^2 + y^2-2yy_c \cosh \eta +y_c^2=0 \}
\end{align}
corresponding to an Euclidean circumference of radius $y_c \sinh \eta$ and center $(x_c,  y_c \cosh \eta)$.
The infinitesimal hyperbolic area is given by
\begin{align} \label{hyperbolic area}
dA= \frac{dxdy}{y^2}
\end{align}
and the area of the hyperbolic circle is thus
\begin{align} \label{Area cerchio iperb}
|B_{\eta}(C)|= 4 \pi \sinh ^2 \frac{\eta}{2}.
\end{align} 

An equivalent definition of the hyperbolic half-plane can be given in the complex domain as the region $\mathbb{H}_2=\{ z\in \mathbb{C} :\, \text{Im}(z) >0\}$ where the measure of the infinitesimal arc length is given by $\frac{|dz|}{Im z}$.  
This formulation is  convenient to express all isometries in $\mathbb{H}_2$ which are  given  by the   M\"obius group of transformations (see \cite{Bona}): 
\begin{align} \label{Moeb}
\mathcal{M}:\mathbb{H}_2 \to \mathbb{H}_2: \  \mathcal{M}(z)= \frac{a z+b}{c z +d} \quad \text{with} \ a,b,c,d \in \mathbb{R} ; \ ad-bc=1
\end{align}
and will play an important role in the present work.

We can now give a brief description of our model. The free particle  moves along geodesic lines and is reflected by the scatterers. Considering that the hyperbolic velocity is defined as

\begin{align} \label{velocita iperbolica}
c(t)=\frac{ds}{dt}= \frac{1}{y} \sqrt{\biggl (\frac{dx}{dt}\biggr )^2+ \biggl (\frac{dy}{dt}\biggr )^2},
\end{align}
it is assumed that the particle moves at constant hyperbolic velocity $c(t)=c$. This means that the hyperbolic distance run by the particle in a time $\Delta t$ is given by $c \Delta t$. Therefore, an Euclidean observer sees the particle moving with a position-dependent velocity equal to $c\, y$.
Without loss of generality, one can assume $c=1$. 

We introduce the notion of Poissonian distribution of obstacles in the  Poincar\'e hyperbolic half-plane.
The obstacles are hyperbolic balls of radius $r$, whose centers  are distributed according to a spatial Poisson process which is homogeneous in the sense of the measure (\ref{hyperbolic area}), i.e. the mean number of obstacle centers per unit hyperbolic area (denoted as $\lambda$) is uniform in $\mathbb{H}_2$. 
This means that the obstacles are identical and  homogeneously distributed in  respect to the hyperbolic metric,  yet to an Euclidean observer they appear to be smaller and denser when approaching to the $x$-axis (see figure \ref{figura traiettoria}). \\
The main result of this study is the analysis of a Boltzmann-Grad-type limit in which  the hyperbolic radius $r$ of each obstacle decreases to zero and the density $\lambda$ in the hyperbolic setting diverges to infinity, so that the mean free path $(2 \lambda \sinh r)^{-1}$ remains constant. Under this limit, the density of our process converges to the density of some Markovian random flight. Moreover, we prove that this limit random motion is similar to the process analysed by M.Pinsky \cite{Pinsky}, who generalized the well-known Euclidean isotropic transport process to the case of an arbitrary Riemannian manifold.

\section{Random obstacles in the Poincar\'e half-plane.}
\subsection{Poisson random fields in $\mathbb{H}_2$}

Assume that a countable set $\Pi$ of points is randomly distributed on the Poincar\'e half-plane $\mathbb{H}_2$ with rate $\lambda (x,y)$. We say that $\Pi$ is a Poisson  random field in $\mathbb{H}_2$ if:
\begin{itemize}
\item For any appropriate set $\mathcal{S} \subset \mathbb{H}_2$, the random variable $N(\mathcal{S})$, namely the cardinality of $\Pi \cap \mathcal{S}$, has the following  distribution:
\begin{equation}
\Pr (N(\mathcal{S})=k)= e^{-\Lambda ( \mathcal{S} )} \frac{(\Lambda ( \mathcal{S} ))^k}{k!}
\end{equation}
with 
\begin{align}
\Lambda (\mathcal{S})=  \int _{\mathcal{S}} \lambda (x,y) \frac{dxdy}{y^2}.
\end{align}
\item For any couple of disjoint regions $\mathcal{S}_1$ and $\mathcal{S}_2$, the random variables  $N(\mathcal{S}_1)$ and $N(\mathcal{S}_2)$ are stochastically independent.
\end{itemize}
We here restrict our attention to the case where the rate $\lambda $ is constant (homogeneous hyperbolic Poisson field). Thus, the number of points inside any set $\mathcal{S}\subset \mathbb{H}_2$ has Poisson distribution with parameter $\lambda |\mathcal{S}|$, where 
\begin{align}
|\mathcal{S}| =\int _{\mathcal{S}} \frac{dxdy}{y^2}
\end{align}
is the hyperbolic area of $\mathcal{S}$. 

Therefore the probability to  have exactly $n$ points in a region $\mathcal{S}$ and to find them inside the hyperbolic elements $dc_1, dc_2,... dc_n$ around $c_1, c_2,... c_n$ is given by
\begin{align} \label{poisson}
\Pr \{P_1 \in dc_1 \dots P_n \in dc_n, N(t)=n \}=     \lambda ^n e^{-\lambda |\mathcal{S}|} dc_1....dc_n
\end{align}
 where
\begin{align}
dc_j=\frac{dx_j dy_j}{y_j ^2}.
\end{align}
It is important to observe that the homogeneous Poisson random field only depends on the measure of areas and therefore it is invariant under the group of isometries of $\mathbb{H}_2$ expressed in (\ref{Moeb}).

To have a more complete description of the homogeneous hyperbolic Poisson field, we treat briefly the distributions of the nearest neighbours points. Let us fix a point $O \in \mathbb{H}_2$  and denote by  $T_k$ the hyperbolic distance between $O$ and the $k^{th}$ nearest point of $\Pi$.

Denoting by $B_{\eta}$ the hyperbolic ball of radius $\eta$ and by $d B_{\eta}$ the infinitesimal anulus of radii $\eta$ and $\eta + d\eta$, we have

\begin{align}
\Pr \lbrace T_k \in d \eta \rbrace &= 
\Pr \lbrace N(B_{\eta})=k-1 \rbrace  \Pr \lbrace N (d B_{\eta})=1 \rbrace \notag \\
&= e^{-\lambda |B_{\eta}|}\frac{(\lambda |B_{\eta}|)^{k-1}}{(k-1)!}\, \lambda |d B_{\eta}|
\qquad k\geq 1, \eta >0 \end{align}

Since $|B_{\eta}|=4 \pi \sinh  ^2 \frac{\eta}{2}$, the anulus $ d B_{\eta}$ has measure $2 \pi \sinh \eta d \eta$  and thus

\begin{equation}
\Pr \l T_k \in d \eta \r = e^{- 4 \pi \lambda   \sinh  ^2 \frac{\eta}{2}} \frac{\l 4 \pi \lambda   \sinh  ^2 \frac{\eta}{2} \r ^{k-1}}{(k-1)!} 2 \pi \lambda   \sinh \eta d \eta
\end{equation}

In particular, the distribution for the nearest neighbour $T_1$ reads
\begin{equation}
\label{Raleigh iperbolica}
\Pr \l T_1 \in  d \eta \r = e^{- 4 \pi \lambda  \sinh  ^2 \frac{\eta}{2}}  2 \pi \lambda \sinh \eta d \eta
\end{equation}

with expectation

\begin{align} \label{valore atteso di T1 iperbolico}
 \mathbb{E} \l T_1 \r =   e^{  2 \pi \lambda}  K_0 (2 \pi \lambda),
\end{align}
where 
\begin{align}
K_0 (z) = \int _0 ^{\infty} e^{-z \cosh t}dt
\end{align}
is the modified Bessel function.
Formula (\ref{Raleigh iperbolica}) is the hyperbolic counterpart of the well known Rayleigh distribution, which describes the distance $T^e_1$ of the nearest neighbour point in the case of a Poissonian random field in the Euclidean plane:
\begin{equation}
\Pr \l T^e_1 \in dr \r = e^{- \lambda \pi r ^2}  2 \pi \lambda r dr
\end{equation}
with mean value
\begin{equation} \label{valore atteso di T1 euclideo}
E \l T^e_1 \r= \frac{1}{2\sqrt{\lambda}}.
\end{equation}
By means of the asymptotic formula for the modified Bessel function,  expression (\ref{valore atteso di T1 iperbolico}) reduces to (\ref{valore atteso di T1 euclideo}) for large values of $\lambda $, namely when the expected distance between Poissonian points decreases and an Euclidean description works well. 

\subsection{Poissonian obstacles}

We now introduce the notion of the Poissonian distribution of obstacles into the hyperbolic half-plane $\mathbb{H}_2$, and we distinguish between hard and soft obstacles, which is a common practice  for motions in Euclidean spaces.

 We are inspired by \cite{shot noise}, where the author studies Poissonian soft obstacles in a particular non Euclidean manifold: the surface of a sphere.
 
Let us consider $\Pi$ to   be a homogeneous hyperbolic Poisson field in $\mathbb{H}_2$ with constant intensity $\lambda$ and let us assume that each point $P\in \Pi$  produces a  potential around itself, whose intensity $\phi$ is a function of the geodesic distance  from $P$. It is assumed that $\phi$ is compactly supported, namely $\phi(\eta)=0$ for $\eta>r$.
The hyperbolic ball of center $P$ and radius $r$, where the function $\phi$ is non-null, is known as a soft obstacle. When a particle hits a soft obstacle, it is subject to an interaction described by $\phi$.  Of course, the obstacles may overlap, which occurs whenever the geodesic distance between two Poissonian points is less than $2r$.
Therefore, at a certain point $Q$, the superposition of the action due to the points $P_1...P_N$ located in a hyperbolic ball $B_r(Q)$ defines a new random field 
\begin{equation} \label{shot noise}
V(Q)= \sum_{j=1}^N \phi \bigl (d_h (P_j Q) \bigr ) 
\end{equation}
where $N$ has Poisson distribution with parameter $\lambda |B_r(Q)|$ and $d_h \bigl (P_j Q \bigr )$ is the geodesic distance between $P_j$ and $Q$. 
Two facts play  fundamental roles. The first  is that the random field (\ref{shot noise}) is homogeneous, meaning that the distribution of $V(Q)$ does not depend on $Q$. 
The second is that (\ref{shot noise}) is isotropic, namely the covariance between $V(Q)$ and  $V(Q')$ only depends on the geodetic distance between $Q$ and $Q'$. For the sake of brevity we omit a complete proof of these facts which can easily be obtained by following the same steps as in \cite{shot noise}.

Hard obstacles are hyperbolic disks of radius $r$ centered at the points of a Poisson random field in $\mathbb{H}_2$. 
In many models of random motions, hard obstacles represent  totally absorbing traps with random locations.  In other models, they act as totally reflecting barriers and can be deemed to be the limiting case of soft obstacles where the following intensity function is considered:
\begin{align*}
\phi(\eta)= \begin{cases}  \infty \qquad &\eta \leq r \\ 
0 \qquad &\eta >r      \end{cases}
\end{align*}

Thus, in what follows, we consider a system of hard obstacles, whose centers are distributed according to a homogeneous hyperbolic Poisson field of constant intensity $\lambda$. A configuration of this kind is homogeneous and isotropic as already explained above.

\section{The Lorentz Process in the Poincar\'e half- plane}
\subsection{Description of the model.}

Let us now consider the following mechanical model. A single particle moves in the Poincar\'e half-plane, where static circular obstacles are distributed according to a Poisson measure. At each instant $t$, the state of the particle is described by the couple $(q,v)$, where $q=(x,y)\in \mathbb{H}_2 $ is the position in the half-plane, and $v=(\cos \alpha, \sin \alpha)$  represents the direction of  motion. Whenever the position $q$ of a particle lies outside the  obstacles, the particle moves along the (unique) geodesic line  tangent to $v$ at the point $q$. We assume that the particle has unit  hyperbolic speed, namely the hyperbolic distance traveled in a time $t$ is equal to $t$. For any initial state $(q,v)$ at $t=0$, the evolution of the particle position until the first collision is given by the geodesic flow $\Phi ^{(q,v)}(t)$, for $t\geq 0$. The explicit expression of the geodesic flow is not essential now and will be given in  the Appendix (formula (\ref{evoluzione senza collisioni})).

When a collision with an obstacle occurs, the particle is reflected on its surface.  

In our model we assume that the particle performs  a "specular reflection" in $\mathbb{H}_2$. Now, in order to generalize the notion of specular reflection from $\mathbb{R}^2$ (where it is straightforward) to $\mathbb{H}_2$, we recall these two basic facts.

%By analogy with elastic collisions of the hard-sphere model, we assume that "specular reflections" occur. In a non-euclidean context, the notion of specular reflection is quite intuitive but not completely obvious. 

%For a rigorous definition, we remind two basic facts.

 The first one is that the measure of hyperbolic angles in the Poincar\'e half-plane corresponds to the measure performed by an Euclidean observer (this is not true in general for all the models of hyperbolic space, for instance the Klein disk model).  The second one is that the angle between two geodesic lines coincides with the one formed by the corresponding Euclidean tangents at the point of incidence, as well as the angle between a geodesic line and a circle is the one detected by the respective tangent lines.
 
Therefore, we refer to the specular reflection in $\mathbb{H}_2$ in the  following way: denoting respectively by $g_1$, $\gamma$ and $g_2$ the pre-collisional geodesic, the tangent to the obstacle and the post-collisional geodesic, we say that the particle is specularly reflected if the angle between $g_1$ and $\gamma$ is equal to the angle between $g_2$ and $\gamma$, as shown in figure (\ref{figura riflessione}).

\begin{figure}
  \centering
  \includegraphics[scale=0.19]{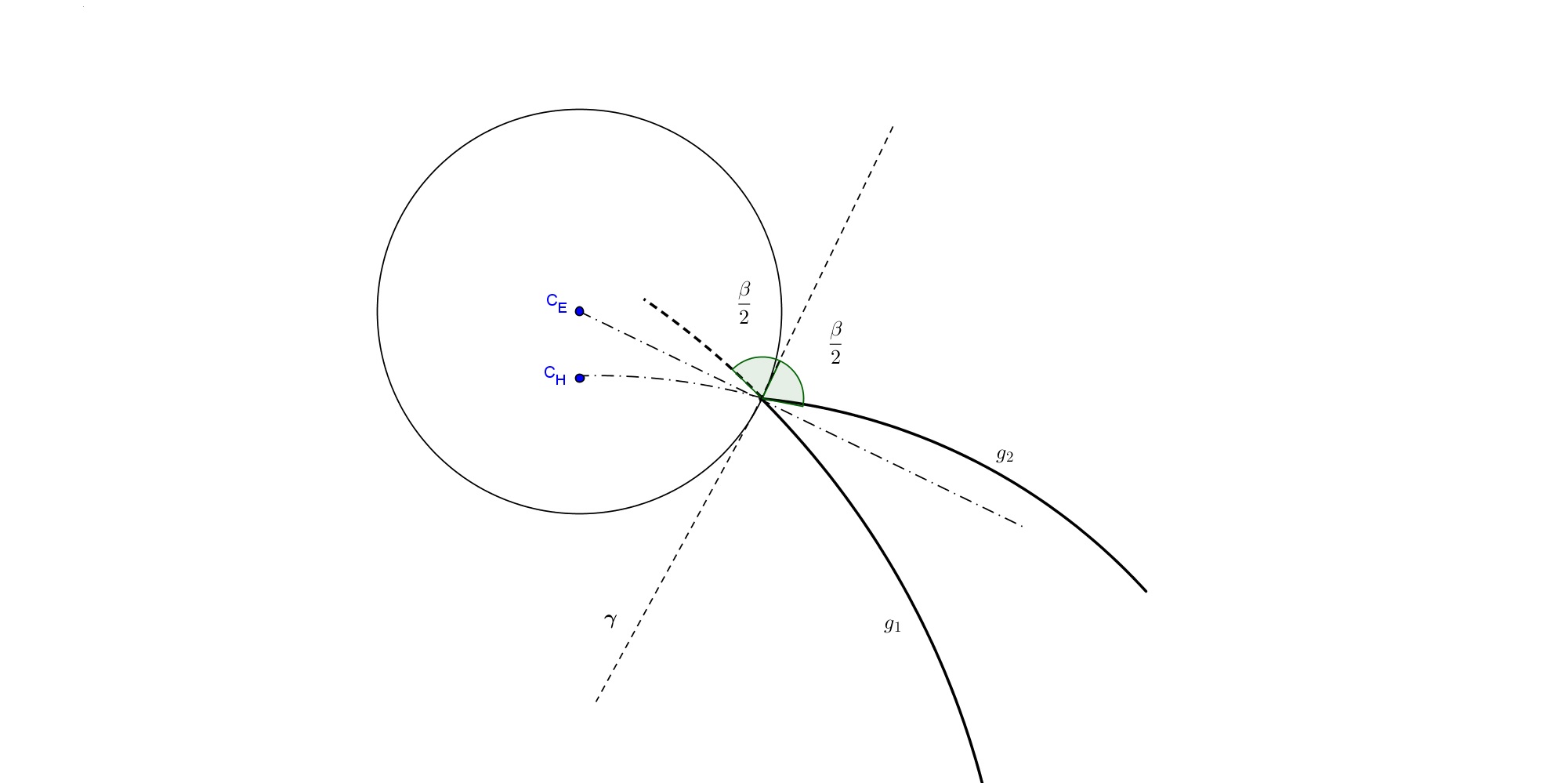}
  \caption{Specular deflection of angle $\beta$ due to an obstacle whose hyperbolic (Euclidean) center is $C_H$ ($C_E$). The angle of incidence and the angle of reflection  are both equal to $\frac{\beta}{2}$.}
  \label{figura riflessione}
\end{figure}

Due to collisions, the sample paths of the moving particle are composed of arcs of circumferences pieced together (for an expression of the piecewise geodesic flow see the Appendix (formula (\ref{evoluzione con le collisioni})).

As a first step, assume now that a configuration of obstacle centers $\{c\}= \{c_1, c_2,... c_j...\}$ is fixed. Let $r$ be the hyperbolic radius of the obstacles.  For a given initial state $(q,v)$, the evolution of the particle position is given by the piecewice geodesic curve $\Phi _{\{c\}}^{(q,v)}(t)$, which clearly only depends on the obstacles of $\{c\}$ centered within a hyperbolic distance $t+r$ from $q$, since the hyperbolic velocity is assumed equal to $1$.  By deriving with respect to $t$ we obtain the Euclidean velocity of the particle (i.e. the velocity perceived by an Euclidean observer) which is denoted by  $\dot{\Phi} _{\{c\}}^{(q,v)}(t)$. The direction of motion is given by the unit vector
\begin{align}
V_{\{c\}}^{(q,v)}(t)= \frac{\dot{\Phi} _{\{c\}}^{(q,v)}(t)}{||\dot{\Phi} _{\{c\}}^{(q,v)}(t)|| }\qquad t\geq 0
\end{align}
where $||.||$ denotes the Euclidean norm.
Hence we define the billiard flow (among the obstacles configuration $\{c\}$) as the following curve on the tangent bundle $\mathbb{H}_2  \times S_1$:
\begin{align} \label{billiard flow}
\mathit{\Psi}_{\{c\}}^{(q,v)}(t)= \bigl ( \Phi _{\{c\}}^{(q,v)}(t), V _{\{c\}}^{(q,v)}(t) \bigr ) \qquad t\geq 0.
\end{align}

By assuming that the locations of the obstacles is random, the evolution of the particle defines a stochastic process $ \{Q_r(t), V_r(t) , t>0 \} $, on  $\mathbb{H}_2 \times S_1$ (the subscript "$r$" representing the radius of the obstacles) that we call hyperbolic Lorentz Process. We denote its joint density by $f_r(q,v,t)$ and suppose that an initial condition $f_r(q,v,0)= f_{in}(q,v)$ is given, such that
\begin{align*}
\int _{\mathbb{H}_2 \times S_1} f_{in}(q,v) dqdv =1.
\end{align*}

 The  function $f_{in}$ should be chosen in such a way that its support lies outside the system of obstacles, but here a difficulty arises since the  obstacles location is random. We can skip this problem by choosing $f_{in}$ as any probability density on $\mathbb{H}_2\times S_1$ and assuming that if the particle initially lies inside an obstacle, it remains at rest forever. It is important to note that such a constraint disappears in the limit of small obstacles considered in this study.
 
   For each $t>0$, the joint density of the hyperbolic Lorentz Process is given by
\begin{align}\label{densita lorentz scritta col valore atteso}
f_r(q,v,t)= \mathbb{E} _{\{c\}} f_{in} \bigl (\mathit{\Psi}_{\{c\}}^{(q,v)}(-t) \bigr )
\end{align}
where the expectation is performed with respect to the Poisson measure.

Before stating the main result of the present work, it is necessary to determine the probability distribution of the free path length among Poissonian obstacles. The calculation requires some properties of hyperbolic geometry, and is shown in detail in the following section.

\subsection{Free path among Poissonian obstacles}
 Let us consider a Poissonian distribution of  spherical  obstacles of hyperbolic radius $r$  in the Poincar\'e half-plane. Suppose that a particle, which is initially located at an arbitrary point $q \in \mathbb{H}_2$, is shot towards an arbitrary direction $v$ and moves  along the geodesic line tangent to $v$ at $q$. 
We are interested in the probability distribution of the first hitting time $T_{(q,v)}$  with the system of obstacles. 
Obviously, under the assumption of unitary hyperbolic speed,  $T_{(q,v)}$ coincides with the  free path length, namely with the hyperbolic distance traveled by the particle without having collisions.

The main idea is the following: the free path $T_{(q,v)}$ is greater than  $t$  if and only if none of the obstacles has its center in the tube
\begin{align}  
\theta (q,v,t)= \biggl \{ p \in \mathbb{H}_2: \inf_ {s \in [0,t]} d_h \bigl (\, p,  \Phi ^{(q,v)}(s)  \bigr)<r \biggr \} 
\end{align}
where $d_h (p,w)$ is the hyperbolic distance between two points $p$ and $w$ of the hyperbolic plane.

In the Euclidean case, the tube is simply given by the union of a rectangle of sides $2r$ and $t$ and two half-circles of radius $r$ (see figure \ref{tubo euclideo}). Some difficulties arise in $\mathbb{H}_2$, where, surprisingly, the two curves at hyperbolic distance $r$ on either side of a geodesic line are not geodesic lines. 

We have then to determine the shape and the hyperbolic area of $ \theta (q,v,t)$. To this aim we use the following representation:  
\begin{align}
\theta (q,v,t)= \bigcup_{0\leq s \leq t} B_r(\Phi^{(q,v)}(s)).
\end{align}
In order to  do that we make use of a suitable transformation in $\mathbb{H}_2$.
Indeed, one can show (see \cite{Bona}, Lemma 2.6) that among M\"obius transformations (\ref{Moeb}) there exists a bijective isometry $\mathcal{M} _{(q,v)}: \mathbb{H}_2 \to \mathbb{H}_2 $
such that\footnote{ We omit the calculations for sake of brevity.  $\mathcal{M} _{(q,v)}$ depends on $(q,v)$ as parameters; for simplicity  we will use the notation $\mathcal{M}$ in the following.} :
 \begin{align} \label{Definition Gamma}
 \mathcal{M}( \Phi^{(q,v)}(s) ) = \Phi^{(\tilde{q},\tilde{v})}(s) \qquad   \forall s \in \mathbb{R}
 \end{align} where  $\tilde{q}=(0,1)$ and $\tilde{v}=(0,1)$, consequently $\Phi^{(\tilde{q},\tilde{v})}(s)=(0,e^{s})$.  
In other words, $\mathcal{M} $ maps any geodesic line into a vertical geodesic line. 
\begin{figure}
  \centering
  \includegraphics[scale=0.19]{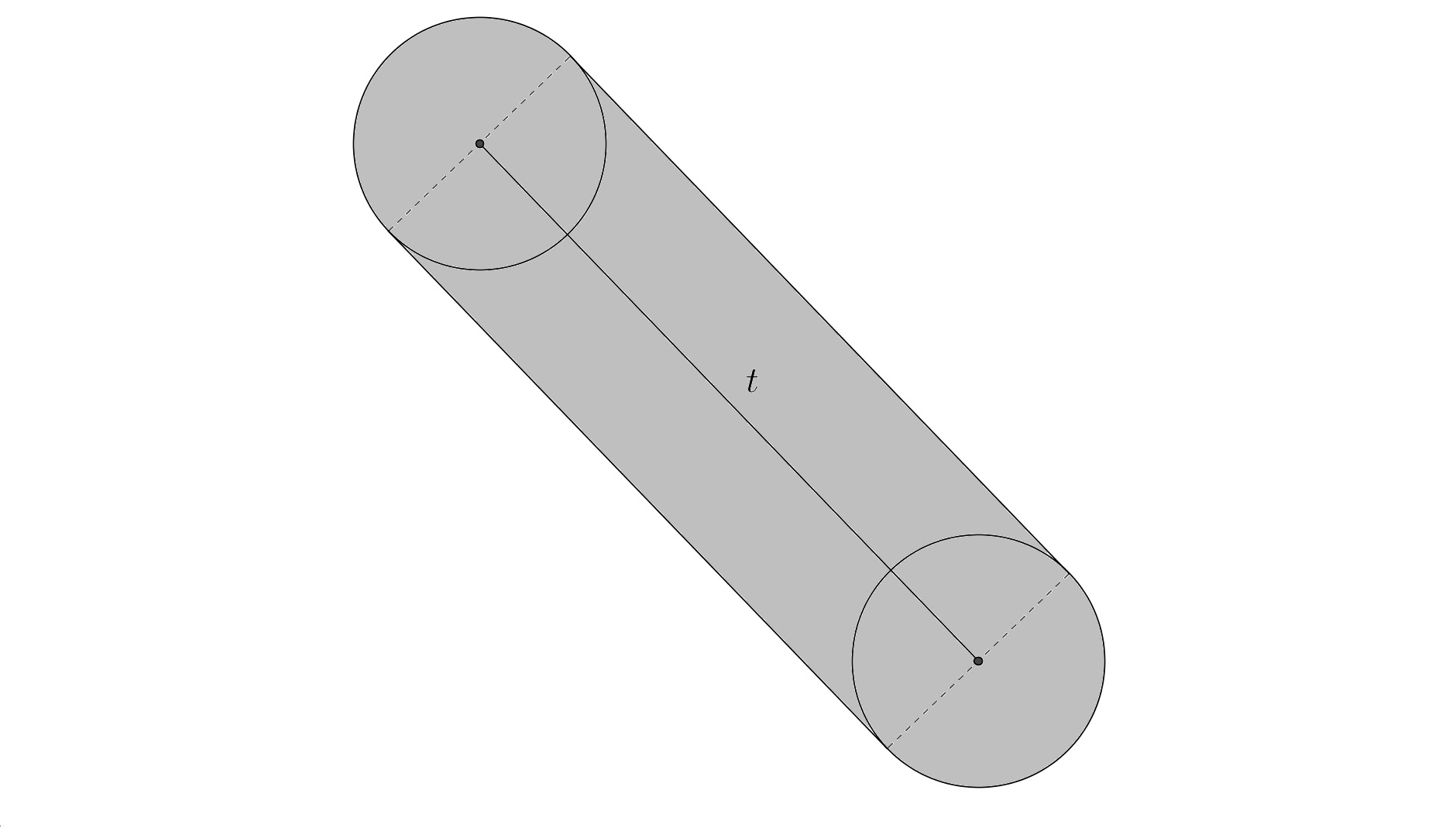}
  \caption{The Euclidean tube-like region $\theta (q, v,t)$ around the free particle trajectory.}
  \label{tubo euclideo}
\end{figure}

Through $\mathcal{M}$ the region $\theta (q,v,t)$ is mapped into the region:
\begin{align}
&\mathcal{M} ( \theta(q,v,t) ) = \biggl \{ w \in \mathbb{H}_2: \inf_ {s \in [0,t]} d_h \bigl ( \,  \mathcal{M}^{-1}(w), \Phi ^{(q,v)}(s) \bigr)<r \biggr \} = \notag \\
& = \biggl \{ w \in \mathbb{H}_2: \inf_ {s \in [0,t]} d_h \bigl ( \, w, \Phi^{(\tilde{q},\tilde{v})}(s)  \bigr)<r \biggr \}  = \theta (\tilde{q}, \tilde{v},t)
\end{align}
where we used that $\mathcal{M}$ is invertible, it preserves distances and the property (\ref{Definition Gamma}), so that the mapped region takes the following simple representation: 
\begin{align} \label{tubo di flusso particolare}
\theta (\tilde{q}, \tilde{v},t)= \bigcup_{0\leq s \leq t} B_r \bigl ( (0,e^s) \bigr ).
\end{align}
Moreover since  $\mathcal{M} $ is an isometry, it preserves areas, whence we can finally compute the desired area as: 
\begin{align} \label{Areas of thetas}
|\theta (q,v,t)|=|\theta (\tilde{q}, \tilde{v},t)|
\end{align}

Now, (\ref{tubo di flusso particolare}) is the region inside the envelope of the following family of curves
\begin{align}
C_{t}= \{ \partial B_r \bigl ((0,s) \bigr ), \ \ 1\leq s\leq e^{t} \}
\end{align}
where $\partial B_r \bigl ((0,s) \bigr )$ has cartesian equation 
\begin{align*}
h(x,y,s)=x^2+(y-s\cosh r)^2-s^2 \sinh ^2 r=0.
\end{align*}
We obtain the envelope of $C_{t}$ by means of the following system
\begin{align}
\begin{cases}
h(x,y,s)=0 \\
\frac{\partial}{\partial s} h(x,y,s)=0
\end{cases}
\end{align}
which gives the union of the following Euclidean lines 
\begin{align}
\label{rette di inviluppo}
y= \frac{x}{\sinh r} \qquad  \qquad y=- \frac{x}{\sinh r}
\end{align}
Thus, the tube $\theta (\tilde{q}, \tilde{v},t)$ is the section of a cone with vertex in $(0,0)$ and central axis the line $x=0$, as shown in figure \ref{tubo iperbolico}. \\
It is now important to observe that  (\ref{rette di inviluppo}) is tangent to $\partial B_r(0,1)$ at the points $A=(\tanh r; \frac{1}{\cosh r})$ and $B=(-\tanh r, \frac{1}{\cosh r})$, and also  tangent to $\partial B_r(0,e^{t})$ at the points $C=(-e^{t}\tanh r; \frac{e^{t}}{\cosh r})$ and $D=(e^{t}\tanh r, \frac{e^{t}}{\cosh r})$.\\
Moreover $A$ and $B$ lie on the geodesic line $x^2+y^2=1$, while $C$ and $D$ lie on $x^2+y^2=e^{2 t}$. This makes it clear that $ \theta (\tilde{q}, \tilde{v},t)$ is composed of three parts: the half-circle below the geodesic segment $AB$,  the intermediate region  $\theta ' (\tilde{q}, \tilde{v},t)$ with vertices $A,B,C,D$ and the half-circle above the geodesic segment $CD$. The hyperbolic area of  $\theta ' (\tilde{q}, \tilde{v},t)$   is defined as 
\begin{align*}
| \theta ' (\tilde{q}, \tilde{v},t)|= \int _{\theta ' (\tilde{q}, \tilde{v},t)  } \frac{dx dy}{y^2}
\end{align*}

\begin{figure}
  \centering
  \includegraphics[scale=0.25]{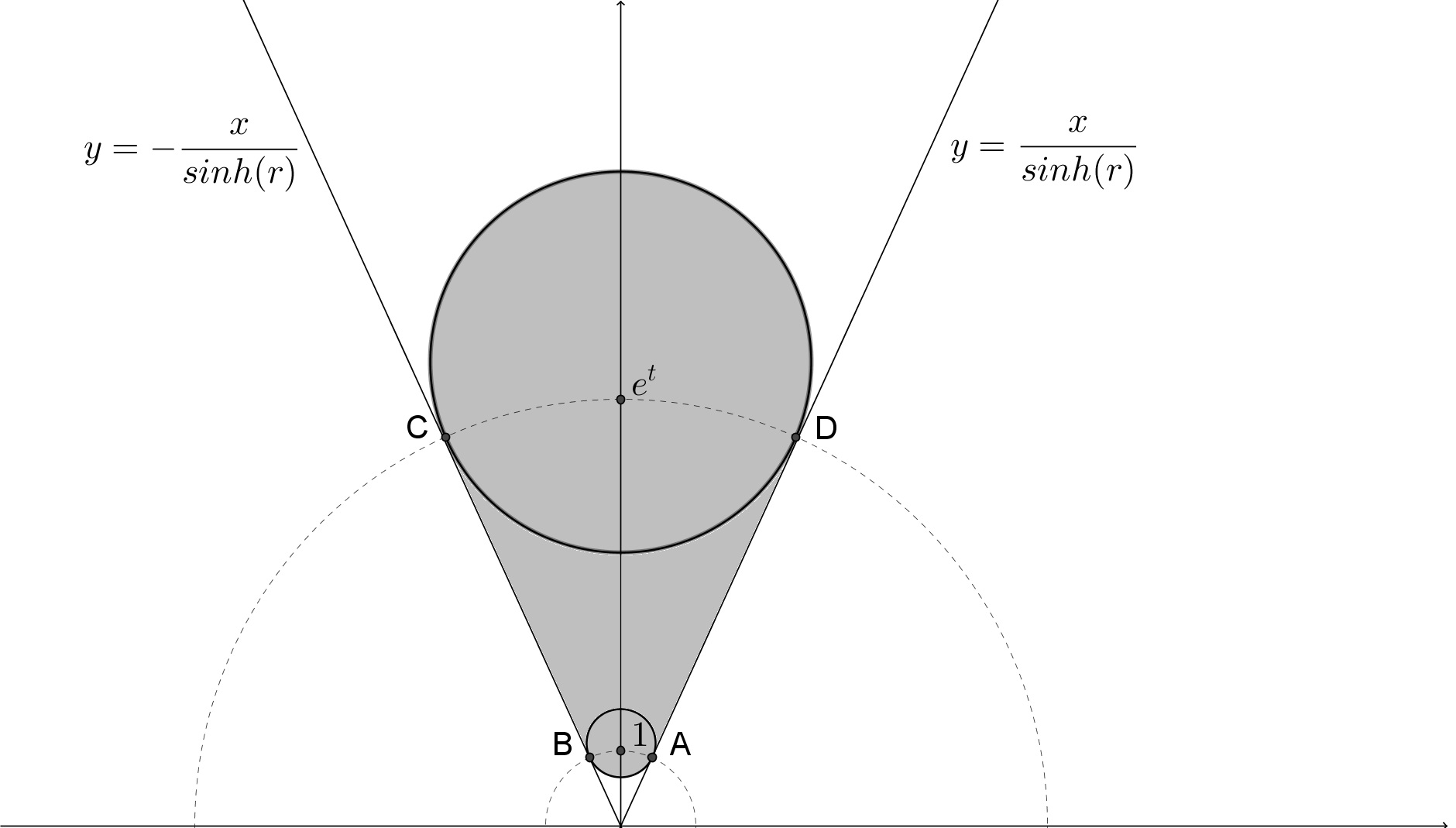}
  \caption{The hyperbolic tube-like region $\theta (\tilde{q}, \tilde{v},t)$ around the free particle trajectory.}
  \label{tubo iperbolico}
\end{figure}

By means of the substitutions $x=\rho\cos \gamma$ and $y=\rho\sin \gamma$ we have that
\begin{align*}
| \theta ' (\tilde{q}, \tilde{v},t)|= \int_1 ^{e^{t}} \int_{\alpha}^{\pi-\alpha} \frac{\rho \,d\rho d\gamma}{(\rho \sin \gamma)^2}= \frac{2 t}{\tan \alpha}
\end{align*}
where $\tan \alpha= \frac{1}{\sinh r}$ is related to the slope of (\ref{rette di inviluppo}). Finally, the area of $\theta (\tilde{q}, \tilde{v},t)$ is given by
\begin{align}
|\theta (\tilde{q}, \tilde{v},t)|= 4\pi \sinh^2 \frac{r}{2}+2t \sinh r
\end{align}

We can now come back to the  probability distribution of the first hitting time $T_{(q,v)}$. As before, we denote by $N(\mathcal{S})$ the number of points inside $\mathcal{S}$. By excluding the possibility that $q$ is located within some obstacles, we thus have
\begin{align*}
&\Pr \ll T_{(q,v)} >t \, | N(B_r(q)) =0  \rr = \frac{\Pr \ll T_{(q,v)} >t \, , \,N(B_r(q)=0\rr }{\Pr \ll N(B_r(q)=0\rr} \\
 &=\frac{e^{-\lambda |\theta (q,v,t)|}} {e^{-\lambda |B_r(q)|}}= e^{-2 \lambda t \sinh r }
\end{align*}
where, in the last equality, we used (\ref{Areas of thetas}) and  (\ref{Area cerchio iperb}).\\
We conclude that $T_{(q,v)}$ has an exponential probability distribution with parameter $2\lambda \sinh r$.
Thus the mean free path, which is a fundamental quantity in what follows, is given by
\begin{align}
\sigma ^{-1} = (2 \lambda \sinh r)^{-1}.
\end{align} 
Performing the same calculation in the Euclidean case leads to say that the free path has an exponential distribution of parameter $2 \lambda r$ and the mean free path is thus $(2 \lambda r)^{-1}$.

\subsection{The main theorem}

The most important result of the present work is the next theorem, where we find a suitable scaling limit corresponding to small obstacles. Among all  possible  scalings, the one consisting in 
\begin{align*}
r \to 0 \qquad \lambda \to \infty \qquad \textrm{in such a way that}\quad 2\lambda \sinh r = \sigma >0,
\end{align*}
that we call hyperbolic Boltzmann-Grad limit in analogy with Gallavotti's work, ensures a non-trivial approximation for $f_r(q,v,t)$. 

Moreover, as a final result, we will also show (see section 3.4) that  the limit function $f(q,v,t)$ is the probability density of a Markovian process, namely a random flight $\{\bigl (Q(t), V(t) \bigr ), t>0 \}$ on the Poincar\'e half plane.
   
\begin{teorema} 
Let $\{(Q_r(t), V_r(t)), t>0\}$ be the Lorentz process in the Poincar\'e half-plane,  defined in such a way that the obstacles are disks of hyperbolic radius $r$, whose centers are distributed as a hyperbolic homogeneous Poisson field with intensity $\lambda = \frac{\sigma}{2\sinh r}$. Let $f_{in} \in  L_{\infty}(\mathbb{H}_2\times S_1)$ be the initial probability density.  Then, in the limit $r \to 0$, the joint density $f_r$ of the Lorentz process converges in $L_1$ sense to some probability density $f$  for each $t>0$. Moreover $f$ solves the following equation
\begin{align} \label{equazione}
\frac{\partial }{\partial t}f(q,v,t)+ \mathcal{D}f(q,v,t)&= -\sigma  f +\sigma \int _0 ^{2 \pi} f(q, R_{\beta}v,t)  \, \frac{1}{4}\sin \frac{\beta}{2}d\beta\\
 f(q,v,0)& = f_{in}(q,v) \notag,
\end{align}
where $\mathcal{D}$ is the operator of covariant differentiation along the geodesic lines and $R_{\beta}$ is the rotation of an angle $\beta$.
 \end{teorema}

 \begin{proof}
It is possible to write explicitly $f_{r}(q,v,t)$. Of course, $f_r(q,v,t)= f_{in}(q,v)$ if $q$ lies inside any obstacle. The following calculations are made on condition that no obstacle center lies inside the ball of hyperbolic radius $r$ around $q$.

Let us suppose that the particle state at time $t$ is given by  $(q,v)$ and consider the backward trajectory. We denote by $N_{(q,v)}(t)$ the number of collisions which occurred up to time $t$.
From section 3.2, it is clear that 
\begin{align}
\Pr \ll N_{(q,v)}(t)=0 \rr = \Pr \ll T_{(q,v)}>t \rr =e^{-2 \lambda t \sinh r}.
\end{align}
Instead, the probability that the particle  collides exactly $n$ obstacles whose centers are located in the infinitesimal hyperbolic areas $dc_1,....dc_n$ around $c_1,....c_n$ is 
\begin{align}
\Pr \ll N_{(q,v)}(t)=n, C_1 \in dc_1, \dots , C_n \in dc_n \rr = \lambda ^n  e^{-\lambda |\theta _{ \{c \}  } (q,v,t)|}dc_1 \cdots dc_n
\end{align}
where $\theta _{ \{c \}  }(q,v,t) $ is the tube of hyperbolic  width $2r$ around the particle trajectory, namely the tube-like region swept by an ideal obstacle when its center is moved along the path. While in the Euclidean case this region is simply given by a non disjoint union of rectangles, it here has a more complex shape and its hyperbolic area can be estimated as  
\begin{align}
|\theta _{\{c\}}(q,v,t)| = 2 t\sinh r  + o(\sinh r).
\end{align}
The  particle density (\ref{densita lorentz scritta col valore atteso}) can be written as
\begin{align}
f_r(q,v,t)& = f_{in}(\Psi ^{(q,v)}(-t)) e^{-2\lambda t\sinh r}\notag \\
&+\sum _{n=1}^{\infty} \int _{A_{q,v}^n} f_{in} (\Psi ^{(q,v)}_{(c_1 \cdots c_n)}(-t)) \lambda ^n  e^{-\lambda |\theta _{\{c\}}(q,v,t)|}dc_1 \cdots dc_n
\end{align}
where $\mathit{\Psi}$ denotes the billiard flow defined in (\ref{billiard flow}) and   $ A_{q,v}^n $ is the subset of $B_{t+r}^n(q)$, containing all the obstacles configurations such that the backward trajectory with initial state $(q,v)$ collides with the $n$ obstacles centered at $c_1 \cdots c_n$.

Following Gallavotti's proof, we observe that among all the possible configurations of obstacles, there are some such that the trajectory hits each obstacle at most once, and there are others that lead to recollisions. Thus we can split $f_r$ into two components, that we respectively call the "Markovian" and the "recollision" terms:
\begin{align}\label{scomposizione in parte Markoviana e non}
f_r(q,v,t) = f^M_{r}(q,v,t)+ f^{REC}_{r}(q,v,t).
\end{align}
We now restrict our attention to the Markovian term $f_r^M(q,v,t)$. By considering the backward evolution, let $\tau _1 \cdots \tau _n$ be the collision times, such that
\begin{align}
0< \tau _1 < \tau _2 \cdots <\tau _n<t 
\end{align}
 and  $\beta _1 \cdots \beta _n$ be the corresponding deflection angles. Then, for a given particle trajectory, there is a one-to one correspondence between  the $2n $ variables $c_1 \cdots c_n$ and the  $2n$ variables $\tau _1 \cdots \tau _n, \beta _1 \cdots \beta _n$.

\begin{figure}
  \centering
  \includegraphics[scale=0.28]{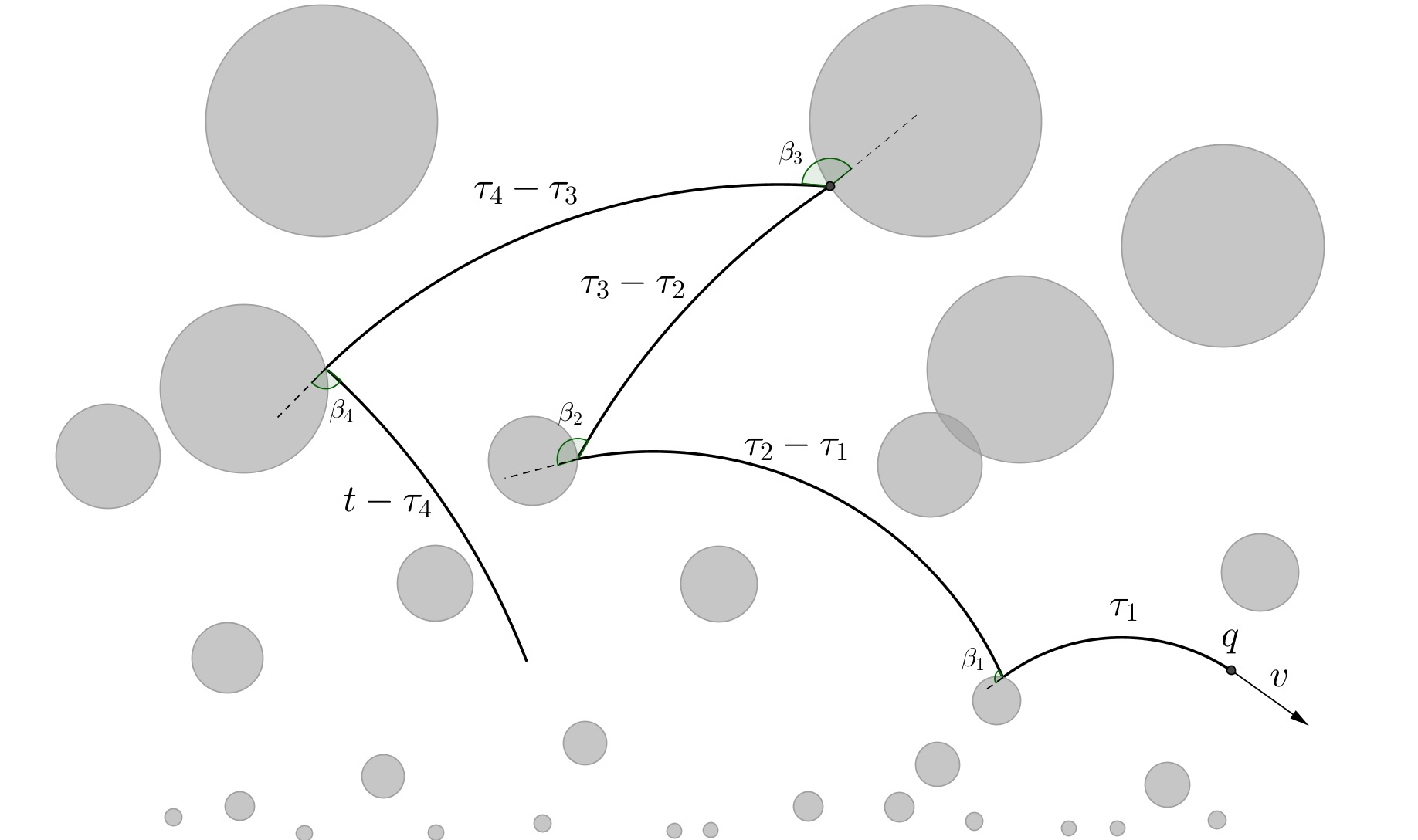}
  \caption{A typical backward trajectory with four collisions, starting from the state $(q,v)$ at time $t$.}
  \label{figura traiettoria}
\end{figure}

We now  obtain the transformation
\begin{align}
dc_1... dc_n = \frac{1}{2^n} \sinh ^n (r) \sin \frac{\beta _1}{2}... \sin \frac{\beta _n}{2} d\tau _1...d\tau _n d\beta _1... d\beta _n
\end{align}
which is a crucial point to prove the theorem.
For simplicity, we treat the case $n=1$, the general case  follows immediately by carrying out tedious calculations.  We explain how to determine the center  $C=(x_C, y_C)$ of the obstacle in terms of the flight time $\tau$ and the deflection angle $\beta$.

Let the particle be at position  $q$ and  unit velocity  $v$  at time $t$; we are interested in its  backward evolution . After a time $\tau$ the particle collides with the first obstacle, namely a circle centered at $C=(x_C, y_C)$, and it is reflected  in a trajectory forming an angle  $\beta$  with the incoming trajectory  as shown in figure \ref{figura riflessione}. \\
In order to perform the calculation we first employ the M\"obius transformation (\ref{Definition Gamma}).
In this way the particle backward trajectory is mapped into 
$\Phi^{(\tilde{q},\tilde{v})}(s)=(0,e^{s})$ for $s \in [0,\tau)$; besides,  the obstacle centered at $C=(x_C, y_C)$ is mapped into an obstacle centered at
 $\tilde{C}=(\tilde{x_C}, \tilde{y_C})$ of the same hyperbolic radius\footnote{Since  $\mathcal{M}$ preserve distances, it sends circumferences into circumferences.}.\\
The traveled time $\tau$ is preserved under an isometry, thus the  point of impact is mapped into $(0, e^{\tau})$; the deflection angle $\beta$ is also preserved since  $\mathcal{M}$ is  conformal (as it is an isometry). These facts will be of great importance in the following calculation.\\
Finally we are interested in the area element $dc$ that can be computed as:
\begin{align}
dc= d \tilde{c}= \frac{d \tilde{x_C} \, d \tilde{y_C} }{\tilde{y_C}^2 }
\end{align}

In this setting, as can be clearly seen in figure  \ref{figura riflessione}, one can  ideally "reach" the obstacle center $\tilde{C}=(\tilde{x_C}, \tilde{y_C})$ from the collision point $(0,e^{\tau})$; it is sufficient to rotate the unit velocity of an angle $\gamma= \frac{\pi}{2}- \frac{\beta}{2}$ and to travel along a path of hyperbolic length $r$. Thus, by using  (\ref{evoluzione senza collisioni}) we have

\begin{align}
\tilde{x_C}& = -\frac{e^{\tau } \sinh r \cos \frac{\beta}{2}}{\cosh r - \sin \frac{\beta}{2} \sinh r}\\
\tilde{y_C}& = \frac{e^{\tau}}{\cosh r - \sin \frac{\beta}{2}\sinh r}.
\end{align}

The Jacobian determinant of the trasformation $(\tilde{x_C}, \tilde{y_C}) \to (\tau, \beta)$ is given by 
\begin{align}
\det J= \frac{1}{2} \frac{e^{2 \tau} \sinh r \sin \frac{\beta }{2} }{(\cosh r - \sin \frac{\beta}{2} \sinh r)^2}
\end{align}
and therefore the  infinitesimal surface element is 
\begin{align}
d\tilde{c}= \frac{d \tilde{x_C} \, d \tilde{y_C} }{\tilde{y_C}^2 }= \frac{1}{2} \sinh r \, \sin \frac{\beta}{2}\,  d\beta d\tau .
\end{align}
as desired.  Thus the Markovian term can be then written as

\begin{align}
f^M_{r}(q, v,t)& =  f_{in}(\Psi ^{(q,v)}(-t)) e^{- 2 \lambda t \sinh r  }+ 
\sum _{n=1}^{\infty} \int _ 0^t d \tau _1 \int _{\tau _1}^t d \tau _2 .... \int _{\tau _{n-1} }^t d \tau _n \notag \\
& \int _{[0,2\pi]^n}d\beta _1 .... d \beta _n \,  f_{in} (\Psi _{\tau _1 \dots \tau _n , \beta _1 \dots \beta _n}^{(q,v)}(-t)) \notag \\
& e^{-\lambda |\theta _ {\{c\}} (q,v,t)|}
  \frac{\lambda ^n}{2^n}  \sinh ^n r \sin \frac{\beta _1}{2}... \sin \frac{\beta _n}{2}
\end{align}.

Observe that the general term in the summation is bounded by
\begin{align*}
||f_{in}||_{L_{\infty}} \frac{(\sigma t)^n}{n!}
\end{align*}
which is the $n^{th}$ term of an exponential converging series.

Thus, in force of the dominated convergence theorem for the series, passing to the limit as  $r \to 0$, together with the assumption $\lambda = \frac{\sigma}{2\sinh r}$, we have that $f_r^M(q,v,t)$ converges pointwise to 

\begin{align} \label{densita limite}
f(q, v,t)= & f_{in}(\Psi ^{(q,v)}(-t))  e^{- \sigma \, t}+ 
 e^{-\sigma t}\sum _{n=1}^{\infty} \sigma ^n \int _ 0^t d \tau _1 \int _{\tau _1}^t d \tau _2 .... \int _{\tau _{n-1} }^t d \tau _n \notag \\
&  \int _{[0,2\pi]^n}d\beta _1 .... d \beta _n \,  \frac{1}{4^n} \sin \frac{\beta _1}{2}... \sin \frac{\beta _n}{2}  f_{in} (\Psi _{\tau_1... \tau _n, \beta _1...\beta_n} ^{(q,v)}(-t))  
\end{align}
and the limit clearly holds also in $L_1$ sense.

By easy calculations we can  check  that (\ref{densita limite}) is a probability density, since
\begin{align} \label{integrazione fm}
||f||_{L_{1}}=\int _{\mathbb{H}_2 \times S_1} f (q,v,t) dqdv=1
\end{align}
and furthermore it is bounded in the following way:
\begin{align} \label{f_M limitata}
||f||_{L_{\infty}} \leq ||f_{in}||_{L_{\infty}}
\end{align}
From (\ref{scomposizione in parte Markoviana e non}) we have
\begin{align} \label{integrale della scomposizione}
\int _{\mathbb{H}_2 \times S_1}f_r(q,v,t)dqdv = \int _{\mathbb{H}_2 \times S_1} f^M_{r}(q,v,t)dqdv+\int _{\mathbb{H}_2 \times S_1} f^{REC}_{r}(q,v,t)dqdv
\end{align}

Passing to the Boltzmann-Grad limit, it is clear that the left hand side of (\ref{integrale della scomposizione}) is equal to $1$, because the billiard flow obviously leaves the total probability mass invariant. Moreover,  conditions (\ref{integrazione fm}) and (\ref{f_M limitata}) and the dominated convergence theorem lead to 

\begin{align}
1=1+ \lim _{r \to 0} \int _{\mathbb{H}_2 \times S_1} f^{REC}_{r}(q,v,t)dqdv
\end{align}
The contribution of the recollision term thus vanishes in $L_1$-norm, namely the measure of all the pathological paths goes to zero. In the end, by writing (\ref{scomposizione in parte Markoviana e non}) as
\begin{align}
f_r(q,v,t)-f(q,v,t)= f_r^M(q,v,t)-f(q,v,t)+f_r^{REC}(q,v,t)
\end{align}
 immediately follows that
\begin{align}\label{parte finale}
||f_r(q,v,t)-f(q,v,t)||_{L_1} \leq ||f_r^M(q,v,t)-f(q,v,t)||_{L_1}+||f_r^{REC}(q,v,t)||_{L_1}
\end{align}
Under the Boltzmann- Grad limit, the right side of (\ref{parte finale}) vanishes and the proof of convergence is complete.

To prove that (\ref{densita limite}) is a solution to (\ref{equazione}), we need to define two kinds of operators, which are bounded in the norm $||.||_{L_{\infty}}$. The first one is the semigroup of geodesic transport with damping, given by
\begin{align*}
T_t f(q,v) = e^{-\sigma t} f(\Phi^{(q,v)}(-t), V^{(q,v)}(-t))
\end{align*}
which is generated by
\begin{align*}
A= -\sigma -\mathcal{D} 
\end{align*} 
where $\mathcal{D}$ denotes the operation of differentiation along the curve $(\Phi ^{(q,v)}(t), V^{(q,v)}(t))$ lying in $ \mathbb{H}_2 \times S_1$.
The second one is the collision operator:
\begin{align*}
L f(q,v)= \sigma \int _0^{2\pi} f(q,R_{\beta}v)\frac{1}{4} \sin \frac{\beta}{2} d \beta.
\end{align*}
Thus (\ref{densita limite}) can be written as
\begin{align*}
f(q,v,t)= T_t f_{in}(q,v)+ \sum _{n=1}^{\infty} \int _{0 \leq s_1 \leq ... \leq s_n \leq t} T_{t-s_n}LT_{s_n-s_{n-1}}...LT_{s_1}f_{in}(q,v) ds_1...ds_n
\end{align*}
which is the Duhamel expansion giving (see \cite{kolokoltsov}, page 52) the solution to
\begin{align*}
\frac{\partial}{\partial t} f=(A+L)f
\end{align*}
and this concludes the proof.
\end{proof}

\subsection{The limit Markovian random flight.}

We now show that a Markovian transport process $\{ (Q(t), V(t)),t>0 \}$ exists whose finite one-dimensional distribution is given by (\ref{densita limite}). The rigorous construction of a process of this kind can be carried out by following the same steps as Pinsky \cite{Pinsky}, who defined a random flight on the tangent
bundle of a generic Riemannian manifold. We assume that $Q(0)=\tilde{q}$ and $V(0)=\tilde{v}$ almost surely, the case of distributed initial data is an immediate consequence.   We denote such a  process by $\bigl (Q^{	(\tilde{q},\tilde{v})}(t), V^{(\tilde{q},\tilde{v})}(t) \bigr )$ and we outline its construction.
Let us consider a sequence of independent waiting times $e_j$, $j\geq 1$, having distribution
\begin{align*}
\Pr \{ e_j>\eta \}= e^{-\sigma \eta} \qquad \eta >0.
\end{align*}
 A particle moves along the geodesic lines in $\mathbb{H}_2$ and changes direction at Poisson times
\begin{align}
\tau _n = e_1+e_2+ \cdots e_n \qquad n \geq 1.
\end{align}

For $0\leq t \leq \tau _1$ we have
\begin{align}
Q^{(\tilde{q},\tilde{v})}(t)= \Phi ^{(\tilde{q},\tilde{v})}(t) \qquad \qquad V^{(\tilde{q},\tilde{v})}(t)= \frac{\dot{\Phi}^{(\tilde{q},\tilde{v})}(t)}{||\dot{\Phi}^{(\tilde{q},\tilde{v})}(t)||},
\end{align}
where $\Phi$ is the geodesic flow defined in (\ref{evoluzione senza collisioni}).
The random point where the first deflection occurs and the corresponding  post-collisional velocity are respectively given by
\begin{align}
Q_1= \Phi^{(\tilde{q},\tilde{v})}(\tau _1) \qquad \qquad V_1= R_{(\beta _1)} \bigl [V^{(\tilde{q},\tilde{v})} (\tau _1^-)\bigr ]
\end{align}
where $R_{(\beta)}$ denotes the rotation of an angle $\beta$. Proceeding recursively, the process is such that, for $ \tau _n \leq t \leq \tau _{n+1}$

\begin{align}
Q^{(\tilde{q},\tilde{v})}(t)= \Phi ^{(Q_n,V_n)}(t-\tau _n) \qquad \qquad V^{(\tilde{q},\tilde{v})}(t)= \frac{\dot{\Phi}^{(Q_n,V_n)}(t-\tau _n)}{|\dot{\Phi}^{(Q_n,V_n)}(t-\tau _n)|}
\end{align}
where the random point of the $j^{th}$  deflection  and the $j^{th}$ post-collisional velocity are denoted by
\begin{align}
Q_j= \Phi^{(Q_{j-1},V_{j-1})}(e_j) \qquad \qquad V_j= R_{(\beta _j)} \bigl [V^{(Q_{j-1},V_{j-1})} (e _j^-)\bigr ]
\end{align}
A crucial point of the construction is that the deflection angles $\beta _j$  are independent of the Poissonian times and among themselves; they have common distribution
\begin{align*}
\Pr \{ \theta_j\in d \beta \}= \frac{1}{4}\sin \frac{\beta}{2}d\beta \qquad \beta \in [0,2\pi] 
\end{align*}
and this is consistent with the cross section due to the collision with a hard sphere.

In the case where the process has initial density $f_{in}(q,v)$, the single-time density of $\bigl \{ (Q(t), V(t)),t>0 \bigr \}$ is just given by (\ref{densita limite}) and it can be written as
\begin{align*}
f(q,v,t)= \tilde{T}_t f_{in} (q,v)= \mathbb{E} \{ f_{in} \bigl ( Q^{(q,v)}(t), V^ {(q,v)}(t)  \bigr )  \}
\end{align*}
where $\{\tilde{T}_t\, , t>0\}$ is a strongly continuous contraction semigroup on the space of differentiable and bounded functions on $\mathbb{H}_2 \times S_1$, endowed with the norm $|| . ||_{L{\infty}}$. 
Then, following the same steps of Pinsky \cite{Pinsky}, we can write the generator of $\tilde{T}_t$: the limit density (\ref{densita limite}) satisfies the following linear Boltzmann-type differential equation (obviously coinciding with \ref{equazione})
\begin{align*}
\frac{\partial }{\partial t}f(q,v,t)+ \mathcal{D}f(q,v,t)= \sigma \int _0 ^{2 \pi} \bigl (f(q, R_{\beta}v,t)-f(q,v,t) \bigr ) \, \frac{1}{4}\sin \frac{\beta}{2}d\beta
\end{align*}
where $\mathcal{D}$ denotes the operator of covariant  differentiation along a geodesic line and $R_{\beta}$ is the  rotation of an angle $\beta$ .

%\section{The Lorentz process on the Poincar\'e disk}
%In the present section we show that our results can be  extended without a lot of effort to another model of hyperbolic geometry, namely the Poincar\'e disk. It is the region $\mathbb{D_2}=\{(x,y)\in \mathbb{R}^2,x^2+y^2<1\}$ endowed with the metric
%\begin{align*}
%ds^2= \frac{4(dx^2+dy^2)}{1-x^2-y^2}
%\end{align*}

%Our argument is based on the fact that the Poincar\'e half-plane and the Poincar\'e disk are isomorphic. The isomorphism is given by the Cayley transform:
%A conformal mapping converts the Poincar\'e half-plane in the Poincar\'e disk. Indeed,
%a point $(x,y) \in \mathbb{H}_2$ is mapped into the point $(u,v) \in \mathbb{D_2}$ with coordinates 
%\begin{align*}
%u= \frac{2x}{x^2+(y+1)^2} \qquad v= \frac{x^2+y^2-1}{x^2+(y+1)^2}.
%\end{align*}
%while the inverse mapping reads
%\begin{align*}
%x=\frac{2u}{u^2+(1-v)^2} \qquad y= \frac{1-(u^2+v^2)}{u^2+(1-v)^2}
%\end{align*}
%In particular, the image of the $x$-axis of $\mathbb{H}_2$ is the border of $\mathbb{D_2}$. Moreover, the Cayley transform is conformal, namely it preserves angles and maps geodesic lines into geodesic lines, hyperbolic circles into hyperbolic circles. 
%In particular, if a geodesic line in $\mathbb{H}_2$ is represented by an euclidean half circles with center $(x_0,0)$ and radius $r$, its image is given by an arc of circumference which is orthogonal to the border od $\mathbb{D_2}$, having center at
%\begin{align*}
%\biggl(\frac{2x_0}{x_0^2-r^2+1}, \frac{x_0^2-r^2-1}{x_0^2-r^2+1} \biggr )
%\end{align*}
%and radius $R$ such that $R^2= (\frac{2r}{x_0^2-r^2+1})^2$.

\subsection{Final remarks}

In this paper we studied the Lorentz process and the related Boltzmann-Grad limit on a classical model of hyperbolic geometry, namely the Poincar\'e half-plane.
It is not straightforward to show that our results  hold  in any hyperbolic space. 
Some problems could arise, for example, when calculating  the mean free path, as this is based on the knowledge of the volume of the tube-like regions.

More precisely, it would be interesting to study the Lorentz process on another well-known hyperbolic manifold, namely the Poincar\'e disk $\mathbb{D}_2$, which we recall to be defined as the set $\mathbb{D}_2 = \{ (x,y) \in \mathbb{R}^2: x^2+y^2<1 \}$ endowed with the metric 
\begin{align}
ds^2= 4 \frac{dx^2+dy^2}{(1-x^2-y^2)^2}.
\end{align}
Equivalently, $\mathbb{D}_2$ can be defined as the complex domain $\mathbb{D}_2 = \{ z \in \mathbb{C}: |z|<1 \}$ with infinitesimal arc-length given by $\frac{2|dz|}{1-|z|^2}$.
We believe that the isomorphism between $\mathbb{H}_2$ and $\mathbb{D}_2$ could be employed to study the Lorentz process in this space.
%Moreover, we would like to underline that it is essential that the hyperbolic angles coincide with the angles measured by an Euclidean observe; this assumption holds on the Poincar\'e half-plane, yet it is false, for example, in the celebrated Klein disc model, where a study of a Lorentz process should not be at all easy (for a model of random motions on the Klein disk see, for example, \cite{kelbert}).
%We imagine, instead, that an extension of our construction could be done 
%with a bit less effort    on the celebrated Poincar\'e disk model, which is defined as the set $\mathbb{D}_2 = \{ (x,y) \in \mathbb{R}^2: x^2+y^2<1 \}$ endowed with the metric
 The isomorphism 
is given by the Cayley transform $\mathcal{K} : \mathbb{H}_2 \to \mathbb{D}_2$ that maps a point $ z \in \mathbb{H}_2$  into a point $ w \in \mathbb{D}_2$: 
\begin{align}
w =\frac{i z +1}{z+i}.
\end{align}
%\begin{align}
%u= \frac{2x}{x^2 + (y + 1)^2} \qquad
%v= \frac{x^2 + y^2 -1}
%{x^2 +(y + 1)^2}
%\end{align}
%The transform $\mathcal{K}$ acts sending the $x$ axis of $\mathbb{H}_2$  into the border of $\mathbb{D}_2$.
%In particular, if a 
If a geodesic line in $\mathbb{H}_2$ is represented by an euclidean half circle with center $(x_0,0)$ and radius $r$, its image through $\mathcal{K}$ is given by an arc of circumference which is orthogonal to the border of $\mathbb{D}_2$, having center at
$\biggl(\frac{2x_0}{x_0^2-r^2+1}, \frac{x_0^2-r^2-1}{x_0^2-r^2+1} \biggr )$
and radius $R$ such that $R^2= (\frac{2r}{x_0^2-r^2+1})^2$.

Now, although $\mathcal{K}$ represents a contraction of the half-plane into the unitary disk, it is a conformal transform, namely it leaves angles between geodesic lines invariant  and we saw how the measure of scattering angles plays a fundamental role in the 
study of the process.
The previous observation, together with the fact that $\mathcal{K}^{-1}$ sends geodesic lines in $\mathbb{D}_2$ into geodesic lines in $\mathbb{H}_2$ suggests that a suitable use of the Cayley transform could be the main tool in the study of the Lorentz process in $\mathbb{D}_2$.

We remark that the assumption that hyperbolic angles coincide with the angles measured by an Euclidean observer doesn't hold for example in the  Klein disc model for hyperbolic space (for random motions with branching on the Klein disk see, for example, \cite{kelbert}).
  
  Finally, we   would like to state  that we did not define the most general Lorentz process on the hyperbolic half-plane. One could investigate, for example, the case of randomly moving obstacles (as Desvillettes and Ricci did in \cite{Desvillettes} in an Euclidean context): it would be reasonable to assume that each obstacle moves with a fixed hyperbolic velocity, following a Gaussian distribution. 

Another line of research could be the analysis of the Lorentz model when other boundary conditions are assumed, e.g. the particle could be re-emitted with a stochastic law  instead of being specularly reflected by the obstacles.

\section{Appendix. Piecewise geodesic motion on the hyperbolic half-plane.}

Before writing an explicit expression for the geodesic flow in the Poincar\'e half-plane, we recall the corresponding one in the Euclidean context.  If a particle starts at  $q \in \mathbb{R}^2$ with velocity $v$ and is not subject to collisions, the position at time $t$ is well known to be equal to
\begin{align} \label{moto rettilineo uniforme}
\Phi ^{(q,v)} (t) = q+vt
\end{align}
On the other hand, let $\tau _1 \dots \tau _n$ be the hitting times, with
\begin{align}
0<\tau _1 <\tau _2 < \dots <\tau _n <t
\end{align}
and let $\beta _1 \dots \beta _n$ be the corresponding deflection angles. The position of  the particle at time $t$, that we call $\Phi ^{(q,v)} _n (t)$ for simplicity of notation, can be written as
\begin{align*}
\Phi _n ^{(q,v)} (t) = q+v \tau _1+v_1 \, (\tau _2-\tau _1) + \dots + v_{n-1}(\tau _n-\tau _{n-1}) +v_n(t-\tau_n)
\end{align*}

with
 \begin{align*}
v_{j}= R_{\beta _{j}}v_{j-1} \qquad \qquad v_0=v
\end{align*}
 $R_{\beta _j}$ representing the matrix of rotation of an angle $\beta _j$.

We now consider a particle starting at $q=(x_0,y_0)\in\mathbb{H}_2$ with velocity $v=(\cos \alpha ,\sin \alpha$). Suppose that the particle moves along the geodesic line with hyperbolic velocity of intensity $1$. 
 Then, the position of the particle at time $t$ is given by
\begin{align} \label{evoluzione senza collisioni}
\Phi ^{(q,v)} (t)=\begin{pmatrix}
 x(t) \\  y(t)
\end{pmatrix} = \begin{pmatrix}
 x_0+ y_0 \frac{\sinh t  \cos \alpha}{\cosh t- \sin \alpha \sinh t}  \\
  \frac{y_0}{\cosh t-\sin \alpha \sinh t} \end{pmatrix}
\end{align}
%We recall that $(t,\alpha)$ are known as the hyperbolic coordinates (i.e. the analogue of the polar coordinates) of the point $(x(t), y(t))$ with respect to  $(x_0,y_0)$.

Formula (\ref{evoluzione senza collisioni}) can be obtained by observing that $(x(t), y(t))$ is obviously given by the intersection of $2 $ curves, namely the hyperbolic circle of radius $t$ centered at $q=(x_0, y_0)$ , having equation
\begin{align}
(x-x_0)^2 + (y-y_0 \cosh t)^2= y_0^2 \sinh ^2 t
\end{align}
and the geodesic line tangent to $ v$ at the point $q$, having equation
\begin{align}
(x-x_0-y_0 \tan \alpha )^2+y^2= \frac{y_0^2}{\cos \alpha ^2}.
\end{align}
By deriving (\ref{evoluzione senza collisioni}) with respect to $t$ we obtain the Euclidean velocity (i.e. the velocity perceived by an Euclidean observer):
\begin{align}
\dot{\Phi}^{(q,v)}(t)= \begin{pmatrix}
\dot{x}(t)  \\ \dot{y}(t)
\end{pmatrix} = \frac{y_0}{(\cosh t- \sin \alpha \sinh t)^2} \begin{pmatrix}
\cos \alpha \\ -\sinh t + \sin \alpha \cosh t
\end{pmatrix}
\end{align}
which is a parallel vector field along the curve $\Phi ^{(q,v)}(t)$, whose norm is given by $||\dot{\Phi}^{(q,v)}(t)|| = y(t) $. The unit velocity vector is given by
\begin{align}
V^{(q,v)}(t)= \frac{\dot{\Phi} ^{(q,v)}(t)}{||\dot{\Phi}^{(q,v)}(t)||} = \frac{1}{\cosh t - \sin \alpha \sinh t} \begin{pmatrix}
\cos \alpha \\ -\sinh t + \sin \alpha \cosh t
\end{pmatrix}.
\end{align} 

We now consider the case where the path is a piecewise geodesic. Let $\tau_1 ... \tau _n$ be the hitting times and $\beta _1... \beta _n$ be the corresponding deflection angles.
The particle starts at $q=(x_0, y_0)$ with velocity $v=(\cos \alpha  , \sin \alpha )$, then it travels along the geodesic line until a time $\tau _1 $, when the position is $\Phi^{(q,v)}(\tau_1)$ and the unit  velocity is given by the vector $ V^{(q,v)}(\tau_1^-)$, which changes into $v_1 = (\cos \alpha _1 , \sin \alpha _1 )= R_{\beta _1}V^{(q,v)}(\tau_1 ^-) $. During the time interval $[\tau_{j-1}, \tau _j]$ the velocity evolves from $v_{j-1} $  to $V(\tau_j^-)  $ and, at time  $\tau _j$ it is changed to  $v_j=(\cos \alpha _j, \sin \alpha _j)=R_{\beta _j} V(\tau_j^-)$. By iterating (\ref{evoluzione senza collisioni}) we immediately obtain 

\begin{align}
 \label{evoluzione con le collisioni}
\Phi _n ^{(q,v)}(t)= \begin{pmatrix}
 x(\tau _n) +y(\tau _n)\frac{ \sinh (t-\tau _n) \cos \alpha _n}{ \cosh (t-\tau _n) - \sin \alpha _{n } \sinh (t-\tau _n) } \\
y(\tau _n) \cdot \frac{1}{\cosh (t-\tau _n) - \sin \alpha _n \sinh (t-\tau _n)  } \end{pmatrix}\qquad 
\end{align}

for $ \tau _n <t<\tau_{n+1}$,   where $x(\tau _n)$ and $y(\tau _n)$ can be computed in a recursive way by using (\ref{evoluzione senza collisioni}).
\\

\textbf{Acknowledgments.}
We wish to thank Alberto Di Iorio for drawing the pictures of this paper.
\\
\

\vspace{.5cm}


\begin{thebibliography}{16}
\providecommand{\natexlab}[1]{#1}
\providecommand{\url}[1]{\texttt{#1}}
\expandafter\ifx\csname urlstyle\endcsname\relax
  \providecommand{\doi}[1]{doi: #1}\else
  \providecommand{\doi}{doi: \begingroup \urlstyle{rm}\Url}\fi


\bibitem GG. Basile, A. Nota, F. Pezzotti,
 M. Pulvirenti;
\textit{Derivation of the Fick's Law for the Lorentz Model in a
low density regime}, Comm. Math. Phys. 336, 1607-1636, (2015)
 \bibitem  {Basile} G. Basile, A. Nota, M. Pulvirenti,   \textit{A Diffusion Limit for a Test Particle in a Random
     Distribution of Scatterers}, J. Stat. Phys . 155:1087-1111, (2014)
\bibitem{boldrighini}  C. Boldrighini, C. Bunimovitch, Ya. G. Sinai, \textit{On the Boltzmann Equation for the Lorentz gas}, J. Stat. Phys., 32, 477-501, (1983).
 \bibitem {Bona} F. Bonahon, \textit{Low-Dimensional Geometry: From Euclidean Surfaces to Hyperbolic Knots}, Student Math. Library, Volume 49, New Jersey (2009).
	
\bibitem {cammarota1}V. Cammarota, E. Orsingher, \textit{Hitting spheres on hyperbolic space} .Theory  Probab. Appl., vol. 57, No 3, pp. 560-587, (2012).
\bibitem {cammarota2}V.Cammarota, E. Orsingher, \textit{Cascades of particles moving at finite velocity in hyperbolic spaces,} J. Stat Phys, Vol.133 , 1137-1159, (2008).
\bibitem {cammarota3}V. Cammarota, E. Orsingher \textit{Travelling randomly on
the Poincar\'e half-plane with a Pythagorean compass,}
J.  Stat Phys., Vol.130, 455-482, (2008)

\bibitem{Desvillettes} L. Desvillettes, V.Ricci,   
\textit{The Boltzmann-Grad limit of a stochastic
Lorentz gas in a force field  } Bulletin of the Institute of Mathematics Academia Sinica (New Series)
Vol. 2  No. 2, pp. 637-648 (2007)

\bibitem {Gallavotti} G. Gallavotti, \textit{Rigorous theory of the Boltzmann equation in the
Lorentz gas,} Nota interna n. 358, Istituto di Fisica, Universit\'a di Roma, (1973)
\bibitem {kelbert} M. Kelbert, Yu. M. Suhov, \textit{Branching Diffusions on $H^d$ with variable fission: the Hausdorff dimension of the limiting set}, Theory Probab. Appl., 51(1), 155-167, (2007)
\bibitem{kolokoltsov} V. Kolokoltsov, \textit{Markov processes, semigroups and generators}, De Gruyter, Berlin, (2011)
\bibitem {degregorio} E.Orsingher, A. De Gregorio, \textit{Random motions at finite velocity in a Non-Euclidean space} Adv. Appl. Prob. Vol 39. n.2, 588-611 (2007).
\bibitem {shot noise} E. Orsingher, \textit{Shot Noise Fields on the Sphere,} Bollettino U.M.I., Series VI, Vol.8, pp.477-496, (1984).
\bibitem {Pinsky} M. A. Pinsky, \textit{Isotropic transport process on a Riemannian manifold}, Trans. Amer. Math. Soc. Vol 218, 353-360, (1976). 
 \bibitem {Spohn} H. Spohn., \textit{The Lorentz process converges to a random flight}. Comm. Math. Phys. 60, 277-290, (1978).

\end{thebibliography}
\end{document}